\documentclass[12pt, twoside, letterpaper]{amsart}

\usepackage{amsmath, hyperref, color, amssymb}
\usepackage{srcltx}
\usepackage[normalem]{ulem}

\usepackage{geometry}
\usepackage{amsthm} \usepackage{amsmath} \usepackage{amsfonts}
\usepackage{amssymb} \usepackage{latexsym} \usepackage{enumerate}

\makeatletter
\@namedef{subjclassname@2020}{%
  \textup{2020} Mathematics Subject Classification}
\makeatother

\theoremstyle{plain}%
\newtheorem{Theorem}{Theorem}[section] %
\newtheorem{Lemma}[Theorem]{Lemma}

\theoremstyle{definition}%
\newtheorem{Remark}[Theorem]{Remark} 
\newtheorem{Example}[Theorem]{Example}

\newcommand{\envspace}{\vspace{2mm}}

\def\D{\mathbb{D}}
\def\E{\mathbb{E}}

\def\P{\mathbb{P}}

\def\R{\mathbb{R}}

\def\be{\mathbf{e}}

\def\d{\delta}
\def\e{\varepsilon}

\def\cC{\mathcal{C}}

\def\cF{\mathcal{F}}

\def\cX{\mathcal{X}}

\renewcommand{\bar}[1]{\overline{#1}}

\usepackage{subcaption}
\usepackage{tikz}
\usepackage{latexsym, amssymb, bbm}
\usepackage{amsmath}
\usepackage{mathrsfs}
\usepackage{amsthm}
\usepackage{verbatim}
\usepackage{fancyhdr}
\usepackage{tikz}
\usetikzlibrary{arrows}
\usepackage{caption}
\usepackage{forest}
\usepackage{graphicx}
\usepackage{stmaryrd}
\usepackage{wasysym}

\usepackage{graphicx}

\usepackage{amsmath}
\usepackage{amsthm}

\DeclareMathAlphabet\mathbfcal{OMS}{cmsy}{b}{n}
\renewcommand{\be}{\begin{equation}}
\newcommand{\ee}{\end{equation}}
\newcommand{\bs}{\begin{split}}
\newcommand{\es}{\end{split}}
\renewcommand{\[}{\left[}
\renewcommand{\]}{\right]}
\renewcommand{\(}{\left(}
\renewcommand{\)}{\right)}

\renewcommand{\epsilon}{\varepsilon}
\renewcommand{\e}{\varepsilon}
\newcommand{\nn}{\nonumber}

\renewcommand{\D  }{\Delta}
\renewcommand{\d }{\Delta}

\usepackage{xcolor}
\usepackage{soul}

\title[Fair Pricing and Hedging Under Small Perturbations of the Num{\'e}raire]{Fair Pricing and Hedging Under Small Perturbations of the Num{\'e}raire on a Finite Probability Space}
\author{William Busching}
\address{William Busching, Bowdoin College, 255 Maine St, Brunswick, ME 04011}
\email{wbuschin@bowdoin.edu}

\author{Delphine Hintz}
 \address{Delphine Hintz, Bethany Lutheran College, 700 Luther Dr, Mankato, MN 56001}
 \email{hintzdelphine@gmail.com}

\author{Oleksii Mostovyi}
\address{Oleksii Mostovyi, Department of Mathematics, University of Connecticut, Storrs, CT 06269, United States}
\email{oleksii.mostovyi@uconn.edu}

\author{Alexey Pozdnyakov}
\address{Alexey Pozdnyakov, University of Connecticut, Storrs, CT 06269, United States}
\email{alexey.pozdnyakov@uconn.edu}

\thanks{This paper is a part of an REU project conduced in  Summer 2021 at the University of Connecticut. The work was supported by the National Science Foundation under grants  No. DMS-1950543 and DMS-1848339.}


\begin{document}

\maketitle

\begin{abstract}
We consider the problem of fair pricing and hedging in the sense of \cite{FS89} under small perturbations of the num{\'e}raire. We show that for replicable claims, the change of num{\'e}raire affects neither the fair price nor the hedging strategy. For non-replicable claims, we demonstrate that is not the case. By reformulating the key stochastic control problem in a more tractable form, we show that both the fair price and optimal strategy are stable with respect to small perturbations of the num{\'e}raire. Further, our approach allows for explicit asymptotic formulas describing the fair price and hedging strategy's leading order correction terms. Mathematically, our results constitute stability and asymptotic analysis of a stochastic control problem  under certain perturbations of the integrator of the controlled process, where constraints make this problem hard to analyze.
\end{abstract}

\section{Introduction}
In complete market models, the benchmark arbitrage-free pricing and hedging  approach is based on replication, and it typically results in a unique price of a given security. It is proven in \cite{Geman95} that, in such settings, a change of num\'eraire affects neither pricing nor hedging. This result allows for more efficient pricing methodologies based on various changes of num\'eraire that are particularly evident for pricing and hedging of interest rate derivatives, where completeness of the underlying model is often embedded in the model assumption. We refer to \cite{BrigoMercurio} for more details.

In incomplete markets, the situation is more complicated, in general. While there is still a class of derivative securities that is replicable, and the assertions of \cite{Geman95} apply to them, there are many other contingent claims that are non-replicable and for which even the notions of a price becomes more complicated. As the arbitrage-free price is not unique for non-replicable claims, other pricing methods have been introduced to overcome the non-uniqueness issue. Among them, fair pricing, see \cite{FS89}. It allows for regaining uniqueness of a (fair) price for a wide class of non-replicable contingent claims.

As a num{\'e}raire is present in essentially every financial model, it is important to understand the response of the pricing and hedging methodologies to perturbations of the num{\'e}raire. In this work, we aim at understanding how fair pricing and hedging change under the small perturbations of num{\'e}raire. Working in settings with multiple stocks, we show that for replicable claims, fair pricing and hedging do not change, while for non-replicable claims, we obtain explicit formulas for the first-order corrections for both the fair price and the hedging strategy. We work with a fairly general parameterization of perturbations of the num{\'e}raire. Also, regardless of the exact form of such perturbations, we show the stability of fair pricing and hedging with respect to small changes of num{\'e}raire. 

Mathematically, we study the sensitivity of a solution of a stochastic control problem to small perturbations of the controlled process. In the settings of a finite probability space, we identify the  conditions for the results to hold, which are represented via the invertibility of certain conditional covariance matrices. This condition is closely related to the non-redundancy of a given set of driving stochastic processes. In the core of our computations is a reformulation of this stochastic control problem from the one where we seek an optimal strategy among the ones satisfying the hard-to-deal with self-financing constraints to the one where these constraints become essentially vacuous. Our examples support the main results by showing that, for non-replicable claims, both the fair price and hedging strategy change under a perturbation of the num{\'e}raire, as well as that the change of num{\'e}raire is fairly different from a seemingly related change of interest rate.

Our results complement the ones in \cite{monatStricker}, where the stability of fair pricing is established with respect to perturbations of the claim’s payoff as well as \cite{BiaginiLocal}, by proving first-order corrections to the fair price and hedging strategy for non-replicable claims in discontinuous stock-price settings. 
Historically, the paper \cite{Merton73}  is one of the first to complete a change in num{\'e}raire although the term num{\'e}raire was never formally defined in \cite{Merton73}, and then \cite{Margrabe78} focuses on and clarifies the term num{\'e}raire and its uses. \cite{Geman95} summarizes the history of the num{\'e}raire and gives convincing examples of its usages. We refer to \cite{KarKar21}  for a  recent literature overview and multiple contemporary results involving the concept of num{\'e}raire.

 The remainder of this paper is organized as follows: in Section \ref{secFP},  { {we specify}}  the mathematical model, introduce the notions of fair pricing and hedging, and relate them to the F{\"o}llmer-Schweizer decomposition. In Section \ref{secNum}, we present the notion of num\'eraire and establish some important related results that are subsequently used in Section \ref{secFSNum} to characterize fair pricing and hedging under a new tradable num{\'e}raire. We provide examples in Section \ref{secExam}. In Section \ref{secStab}, we consider the stability of our solution under small perturbations of the num{\'e}raire; in Section \ref{secAsym}, we obtain explicit formulas for the first-order correction terms of each component under small perturbations of the num{\'e}raire.
 
\section{Fair Pricing and Hedging}\label{secFP}
This section aims to discuss the fair pricing and hedging in the settings with multiple stocks, making a digression and considering the case separately with one stock, where the computations become a bit simpler.
Consider an economy indexed over discrete time, with uncertainty represented by a finite probability space $(\Omega, \mathbb{F}, \mathbb{P})$.  The flow of information to all agents in this economy is represented by the filtration $\mathbb{F} = (\mathcal{F}_n)_{n=0,1,\ldots,T}$ with fixed $T \in \mathbb{Z}^+$. Assume that $\mathcal{F}_0$ is trivial, containing only $\emptyset$ and $\Omega$, and that $\mathcal{F}_T$ is the power set of $\Omega$.

We begin by supposing there is a bank account $S^0$ with a price process equal to one at all times. We use this bank account as a num{\'e}raire in our introduction of the F{\"o}llmer-Schweizer decomposition  while noting that our subsequent analysis   enables us to consider a more general tradable num{\'e}raire in the same circumstances.  Let $S = (S_n)_{n=0,1,\ldots,T}$ be a $d$-dimensional vector-valued  $\mathbb{F}$-adapted process; i.e., each $S_n$ is $\mathcal{F}_n$-measurable, and let $S^i$ describe  the evolution of the $i$th stock, $i\in\{1,\ldots, d\}$. We take $S$ to describe the discounted price process of $d$ stocks, and we denote the vector-valued increments of $S$ by $\Delta S_n := S_n - S_{n-1}$ for $n=1,\ldots,T$.

Let $\xi = (\xi_n)_{n= 1,\ldots,T}$ represent a $d$-dimensional trading strategy corresponding to the number of shares of stocks held at any time $n$. 
We   restrict the strategies to the ones that are {\it predictable} and {\it self-financing}, in the following senses: since our position in the stock market at time $n$ must be chosen at time $n-1$, we say that $\xi$ is \textit{predictable}, namely, that $\xi_n$ is $\mathcal{F}_{n-1}$-measurable for each $n\in \{ 1,\ldots,T\}$. Let $\xi^0$ denote the positing in the bank account, and denoting $\bar \xi: = (\xi^0, \xi)$, $\bar S := (1, S)$, we call a strategy $\bar \xi$ to be  \textit{self-financing} if 
\be\label{sf}
\bar \xi_n \cdot  \bar S_n = \bar \xi_{n+1} \cdot  \bar S_n,\quad n \in\{1,\ldots, T-1\},
\ee
where $\cdot$ denotes the scalar product in $\R^{d+1}$; below we also use the same symbol for the scalar product in $\R^d$. 
  \begin{Remark}\label{remSf}
  Condition 
   \eqref{sf} implies that the accumulated gains and losses resulting from the asset price fluctuations are the only sources of changes in the portfolio value. It can be restated in the following equivalent forms
   \be\label{871}
   \bar \xi_{n+1} \cdot \bar S_{n+1} - \bar \xi_{n}\cdot \bar S_n = \bar \xi_{n+1}\cdot (\bar S_{n+1} - \bar S_n),\quad  n \in\{1, \ldots, T-1\}.
   \ee 
   Summation over \eqref{871} implies that
     \be\label{872}
   \bar \xi_{n} \cdot \bar S_{n} = \bar \xi_{1}\cdot \bar S_0 +\sum\limits_{k = 1}^n\bar \xi_k\cdot (\bar S_{k} - \bar S_{k-1}), \quad n   { {=}}1, \ldots, T,
   \ee  
   where $\bar \xi_{1}\cdot \bar S_0$ is the initial wealth. We refer to 
  \cite[p. 293]{FS16} for more details. Below, we use the self-financing condition in the form \eqref{872}.
  \end{Remark}

 Let $\Theta$ be the set of all predictable $d$-dimensional processes (that vacuously correspond to self-financing trading strategies). For $\xi \in \Theta$, we define the gains process 
\be\label{2.1}
    G_n(\xi) = \sum_{j=1}^n \xi_j \cdot\Delta S_j,
    \quad n =   1, \ldots, T. 
 \ee
$G_n$ can be thought of as gains in wealth up to time $n$ for a self-financing trading strategy. Here and below $\d W_n = W_n - W_{n-1}$, for every process $W$. 
An important observation is that, on the right-hand side of \eqref{2.1}, as $S^0\equiv 1$, one can equivalently use $\bar \xi_j\cdot \bar \D S_j$; however, the self-financing condition for such $\bar \xi$'s must hold, in this case. 

Let  $V_0$ denote  the initial capital invested in the market at time $n=0$. Then the total output from the trading process at time $n$ is given by $V_0 + G_n(\xi)$. 
Now suppose there exists  derivative security that pays a value $H$ at the final time $T$. To successfully hedge this security, we wish to trade in such a way that the trading output at time $T$ is as close as possible to the payout of the derivative security.  One way to ensure this is to minimize the expected quadratic cost incurred from hedging the security -- namely, by solving the minimization problem
\begin{equation}\label{minPr}
    \min_{V_0 \in \mathbb{R}, \xi \in \Theta} \mathbb{E}\[\(H-V_0-G_{ { {T}}}(\xi)\)^2\].
\end{equation}
    Following \cite{FS89}, the  $V_0$ in \eqref{minPr} is called {\it fair price}. We also call the optimal $\xi$ the {\it fair price-based hedging strategy}. 
To characterize them, we introduce the discrete F{\"o}llmer-Schweizer decomposition, following \cite{FS89}  and \cite{monatStricker}. 

\begin{Theorem}\label{thm2.2}
Let $S = M + A$ be the semimartingale decomposition of $S$ into a martingale $M$ and a predictable process $A$. Then every square-integrable and $\mathcal{F}_N$-measurable contingent claim $H$ admits a decomposition
\begin{equation}\label{2.4}
    H = V_0 + \sum_{k=1}^T \xi_k\cdot\Delta S_k + L_T 
\end{equation}
for some $V_0 \in \mathbb{R}$, a process $\xi \in \Theta$, and a martingale $L $, where 
\begin{enumerate}
    \item $L $ and every component of $M$ are orthogonal, meaning $\mathbb{E}[\Delta L_n   \Delta M^i_n \mid \mathcal{F}_{n-1}]=0$,  $n\in\{1,\ldots, T\}$, $i\in\{1,\ldots, d\}$;
    \item $\mathbb{E}[L_0] = 0$.
\end{enumerate}
\end{Theorem}
Following this decomposition, successful hedging of a contingent claim $H$ requires the minimization of the unhedgeable $L_T$ term. In turn, this is closely related to the concept of {\it local risk minimization}, see \cite{BiaginiLocal}.
\begin{Remark}
When $S$ is a martingale, the decomposition \eqref{2.4} is the Galtchouk-Kunita-Watanabe decomposition, see \cite{KW} and \cite{Galchouk}.
\end{Remark}




The F{\"o}llmer-Schweizer decomposition is crucial in obtaining optimal solutions $\xi$ and $V_0$ to \eqref{minPr}. Indeed, following F{\"o}llmer and Schweizer \cite{FS89}, for $d = 1$ (that is with one risky asset), the recursive formula for $\xi$ is given by

\begin{equation}\label{2.3}
    \xi_n = \frac{\text{Cov}_{\mathcal{F}_{n-1}} (H - 
    \sum_{k=n+1}^N \xi_k \Delta S_k, \Delta S_n)}{
    \text{Var}_{\mathcal{F}_{n-1}}[\Delta S_n]},\quad n = 1,\ldots, T,  
\end{equation} 
where $\text{Cov}_{\mathcal{F}_{n-1}}$ and $\text{Var}_{\mathcal{F}_{n-1}}$ denote the conditional covariance and variance, respectively. 

It is this optimal strategy that we examine under a change of num{\'e}raire in the next section. Note that, in \eqref{2.3}, we need the process  $\text{Var}_{\mathcal{F}_{n-1}}[\Delta S_n]$, $n\in\{1,\ldots,T\}$, to be {\it strictly} positive. 
\begin{Remark}
With multiple risky assets, this condition is replaced with the invertibility of conditional covariance matrices with probability $1$. This condition is closely related to the non-redundancy of given stocks. However, the stability and asymptotic analysis are only imposed for the base model corresponding to $\e = 0$.  
\end{Remark}


\section{
Change of Num{\'e}raire}\label{secNum}

To establish the machinery to change to a general num{\'e}raire, we begin by defining the set of general wealth processes starting from $x$ to be
\begin{equation}\label{3.1}
    \mathcal{X}(x):= \left\{ x + G_n(\xi) = x+ \sum_{k=1}^n \xi_k \cdot \Delta S_k \mid \xi \in\Theta ,~n\in\{0,\ldots,T\} \right\},\quad x\in\R.
\end{equation}
A {\it num{\'e}raire}, most generally, can be defined as any {\it strictly} positive non-dividend-paying asset. We   focus on a tradable num\'eraire: that is, the ones  where $N$ is a strictly positive element of $\cX(1)$. That is, $N$ 
has the form 
\begin{equation}
\label{3.2}
        N_n = 1+\sum_{k=1}^{n}\eta_k\cdot\Delta{S}_k,\quad n\in\{0, \ldots, T\},
\end{equation}
for some $\eta\in\Theta$. $N_0 = 1$ is a normalization condition that is common in the literature. 

The F{\"o}llmer-Schweizer decomposition discussed above utilizes an unchanging bank account as a num{\'e}raire -- i.e., $N\equiv 1$ is constant at all times $[0, T]$. Indeed, a num{\'e}raire $N\equiv 1$ is implicit in many results in the field. 

By a change in num{\'e}raire, we mean  a change of units  in which a  price process of the traded securities, $\bar S = (1, {S})$, is measured.  Note that this process includes the bank account with value one across our time horizon as well as $d$ stocks - all encoded in a vector $\bar {S}$.  We denote the $(d+1)$-dimensional price process of traded securities under a change of num{\'e}raire $$S^N := \(\frac{1}{N}, \frac{S}{N}\).$$ 

Under the new num\'eraire, it is natural to introduce the set of wealth processes, analogous to \eqref{3.1}; this is done in  \eqref{XN} below. To emphasize the self-financing constraints under a change of num\'eraire, first, we notice that one can rewrite \eqref{3.1} as
\begin{equation}\nn
\bs
    \mathcal{X}(x)= \left\{ x+ \sum_{k=1}^n \bar\xi_k \cdot \Delta \bar S_k \mid \bar \xi~is ~predictable~and~self-financing,~n\in\{0,\ldots,T\}\right\},\\ \quad x\in\R.\end{split}
\end{equation}
This allows us to naturally extend \eqref{3.1} to the set of wealth processes under the num\'eraire $N$ as follows:
\begin{equation}\label{XN}
\bs
    \mathcal{X}^N(x):= \left\{   x+ \sum_{k=1}^n \bar\xi_k \cdot \Delta S^N_k \mid \bar \xi~is ~predictable~and~self-financing,~n\in\{0,\ldots,T\}\right\},\\
    \quad x\in\R,
\end{split}
\end{equation}
where the self-financing condition, analogous to \eqref{sf}, now must hold under the num\'eraire $N$,  that is 
\be\label{sfN}
  \bar\xi_n \cdot  S^N_n =   \bar\xi_{n+1} \cdot    S^N_n,\quad 1 \leq n \leq T-1,
\ee
or, similarly, along the lines of Remark \ref{remSf}, where we   use
     \be\label{873}
   \bar \xi_{n} \cdot   S^N_{n} = \bar \xi_{1}\cdot   S^N_0 +\sum\limits_{k = 1}^n\bar \xi_k\cdot (  S^N_{k} -   S^N_{k-1}), \quad n \in\{1, \ldots, T\},
   \ee  
   and where $\bar \xi_{1}\cdot   S^N_0$ is the initial wealth. 

We begin by noting a convenient result, analogous to   \cite[Lemma 6.1]{MostovyiNumeraire}, which demonstrates that wealth processes under a change of tradable num{\'e}raire   adjust  in an expected way. {\it In particular,  the replicable claims stay replicable under a change of num\'eraire.} 
The proof of this result is similar to the proof of Lemma \ref{lem872} below, and it is skipped. 

\begin{Lemma}\label{lemCN}
Consider a stock price process under a change of num{\'e}raire $S^N = (\frac{1}{N}, \frac{S}{N})$. 
 { {Let us fix} $ {x\in\R}$  {and consider the sets of wealth processes under the old and new num{\'e}raires} $ {\cX(x)}$  {and} $ {\cX^N(x)}$  {given by} \eqref{3.1}   {and} \eqref{XN},  {respectively}}. 
Then, we have 
\begin{equation}\label{3.3}
    \mathcal{X}^N(x) = \frac{\mathcal{X}(x)}{N} = \bigg\{ \frac{X}{N} = \(\frac{X_n}{N_n}\)_{\{n\in\{0,\ldots, T\}\}}\mid X \in \mathcal{X}(x) \bigg\}.
    \end{equation}
\end{Lemma}

 \section{Fair Pricing and Hedging and a Num\'eraire Change}\label{secFSNum}
We now turn to the question of applying the F{\"o}llmer-Schweizer decomposition-based hedging mechanism in an environment with a new num{\'e}raire. This has to be done with care. 
First, let us observe that the objective function in \eqref{minPr} can be rewritten as
\begin{equation}\label{854}
    \min_{X\in\bigcup\limits_{x\in\R} \cX(x)} \mathbb{E}\[\(H-X_T\)^2\].
\end{equation}
As the contingent claim $H$ measured under the new num\'eraire to $N$ is worth $H/N$, using the notations of the previous section, the natural formulation of the \eqref{minPr} under $N$ becomes
\begin{equation}\label{855}
    \min_{X^N\in\bigcup\limits_{x\in\R} \cX^N(x)} \mathbb{E}\[\(\frac H {N_{ {  T}}} -X^N_T\)^2\].
\end{equation}
Recalling the   definition of $\cX^N(x)$'s (including the self-financing condition  \eqref{sfN}), the latter minimization problem \eqref{855} can be restated as follows.
\begin{equation}\label{3.4}\bs
     {\rm minimize}&\qquad \mathbb{E} \left[ \left(\frac{H}{N_T}-\frac{V_0}{N_0}-\sum_{k=1}^T\bar \xi_k \cdot \Delta S^N_k\right)^2\right],\\
     {\rm subject~to}  &\qquad V_0 \in \mathbb{R},\\
    &\qquad \bar\xi ~is~predictable~and~satisfying ~\eqref{sfN}.
    \end{split}
\end{equation}
The optimal $V_0$ and $\bar\xi$ to \eqref{3.4} (whose existence is proven below)  are defined to be the {\it fair price} and the {\it fair price-based hedging strategy} (or simply, the hedging strategy) under the num\'ereire $N$. 
 
One can see that the self-financing constrain \eqref{sfN} enters problem \eqref{3.4} and makes it harder to analyze. 
Whereas in \eqref{minPr}  we considered a risk-free asset, whose increments $\Delta S^0\equiv  0$ at all times as a component of $\Delta \bar S$, $\Delta S^N$ has no zero component, in general, since a tradable num{\'e}raire may change over time. 
In the change of num{\'e}raire case, the optimal $\bar\xi$ then additionally depends on the risk-free asset, which may evolve over time. The component of $\bar\xi$ corresponding to investment in the risk-free asset cannot, therefore, be chosen essentially after solving \eqref{3.4} in a way to  make the optimal $\bar\xi$ self-financing, as the self-financing condition \eqref{sfN} is a constraint on our minimization problem \eqref{3.4}.  The following lemma demonstrates how we can bypass this issue by reformulating the minimization problem \eqref{3.4} in a more convenient way.

 { {Let us set}}
\begin{equation}\label{3.5} { {
    \min_{V_0 \in \mathbb{R},\xi \in \Theta} \mathbb{E} \left[ \left(\frac{H-V_0-\sum_{k=1}^T \xi_k \cdot \Delta S_k}{N_T} \right)^2 \right], }}
\end{equation}
 { {and let}} $ { {W\in\bigcup\limits_{x\in\R}\cX(x)}}$  { {and}}  $ { {W^N\in\bigcup\limits_{x\in\R}\cX^N(x)}}$  { {denote the optimal wealth processes to}  \eqref{3.4}   {and}  \eqref{3.5}  {, respectively.}} 
\begin{Lemma}\label{lem871}
Let $N$ be a  { {tradable}}  num\'eraire. Then,   \eqref{3.4} is equivalent to  \eqref{3.5} in the sense that 
 the objective functions are equal, and there is a one-to-one correspondence between the optimal wealth processes  to \eqref{3.4} and \eqref{3.5}, which are unique with probability $1$, and we have 
\be
\label{owp}
W^N = \frac{W}{N}.\ee
 
\end{Lemma}
\begin{proof}
Observe that the expression $X^N := V_0+\sum_{k=1}^T \bar\xi_k \cdot \Delta S^N_k$ $\in$ $\mathcal{X}^N(V_0)$.  Applying \eqref{3.3}, we have $X^N = \frac{X}{N}$, for some $X$ of the form $X_n = V_0+\sum_{k=1}^n \xi_k \cdot \Delta S_k$, $n\in\{0,\ldots, T\}$, such that $X\in \mathcal{X}(V_0)$. Note that the $\xi$ that gives rise to each wealth process may be different, but $V_0$ must be the same in both, due to the normalization condition, $N_0 = 1$ and Lemma \ref{lemCN}. Then \eqref{3.4}, or rather an equivalent problem \eqref{855}, becomes
    \be\label{171}\bs
        \min_{X^N \in\bigcup\limits_{V_0 \in \mathbb{R}} \mathcal{X}^N(V_0)} 
        \mathbb{E} \left[ \left(\frac{H}{N_T}-X^N_{ { {T}}}\right)^2\right]
        &= \min_{X \in  \bigcup\limits_{V_0 \in \mathbb{R}} \mathcal{X}(V_0)} 
        \mathbb{E}\left[\left(\frac{H}{N_T}-\frac{X_T}{N_T}\right)^2\right] \\
        &= \min_{V_0 \in \mathbb{R},\xi \in \Theta} 
        \mathbb{E} \left[ \left(\frac{H-V_0-\sum_{k=1}^T \xi_k \cdot \Delta S_k}{N_T} \right)^2 \right],
    \end{split}\ee
which  is \eqref{3.5}. The chain of equalities above shows the objective functions in  \eqref{3.4} and \eqref{3.5} are equal. 
Next, using the direct method from the calculus of variation and strict convexity of the function $x\to x^2$, $x\in\R$, appearing in the objective, one can show the existence and uniqueness of the optimal self-financing wealth processes under the corresponding num\'eraires that are  the minimizers to  \eqref{3.4} and  \eqref{3.5}.  The computations above,  { {in particular,} \eqref{171}   {and Lemma}  \ref{lemCN}}, imply \eqref{owp}. 
\end{proof}

To emphasize the self-financing constraints, similarly to reformulation \eqref{3.4}, one can restate \eqref{3.5} as 
\be\label{8101}
\bs
     {\rm minimize}&\qquad \mathbb{E} \left[  \left(\frac{H-V_0-\sum_{k=1}^T \bar\xi_k \cdot \Delta \bar S_k}{N_T} \right)^2\right],\\
     {\rm subject~to}  &\qquad V_0 \in \mathbb{R},\\
    &\qquad \bar\xi ~is~predictable~and~satisfying ~\eqref{sf}.
    \end{split}
\ee
The following lemma establishes a relationship between the optimal hedging strategies for \eqref{3.4} and \eqref{8101}.
\begin{Lemma}\label{lem872}
Let $W^N_n = V_0 + \sum_{k=1}^n \xi^{(1)}_k \cdot \Delta S^N_k$,  $n\in\{0,\ldots,T\}$, and $W_n = V_0 + \sum_{k=1}^n \xi^{(2)}_k \cdot \Delta \bar S_k$, $n\in\{0,\ldots,T\}$, be the optimal self-financing wealth processes for \eqref{3.4} and \eqref{8101}, respectively. Then, we have 
\be\label{878}
\xi^{(1)}_n \cdot \Delta \bar S_n = \xi^{(2)}_n \cdot \Delta \bar S_n,\quad n\in\{1,\ldots, T\},
\ee
\be\label{879}
\xi^{(1)}_n \cdot \Delta   S^N_n = \xi^{(2)}_n \cdot \Delta   S^N_n,\quad n\in\{1,\ldots, T\}.
\ee
In particular, one can use the same strategy to optimize both \eqref{3.4} and \eqref{8101}.
\end{Lemma}
\begin{Remark}
Lemmas \ref{lem871} and \ref{lem872} assert that, for replicable claims (that is, the ones that are represented by a terminal value of an element of $\bigcup\limits_{x\in\R}\cX(x)$) change of num\'eraire affects neither the fair price nor the hedging  strategy. 
\end{Remark}
\begin{proof}[Proof of Lemma \ref{lem872}]
%
%
%
%
    As, by Lemma \ref{lem871}, $N_nW^N_n =  W_n$, $n\in\{0,\ldots, T\}$, we get $$        \d W_n  =\Delta(W^N_n {N}_n) ,$$
    that is  
    \begin{equation}\label{875}
    \bs
         \xi_n^{(2)}\cdot\Delta{\bar S}_n 
        & = W_{n-1}^N \Delta{N}_n+N_{n-1} \Delta{W}^N_n+\Delta{W}_n^N \Delta{N}_n \\
        & = W_{n-1}^N\Delta{N}_n+(N_{n-1}+\Delta{N}_n)\xi_n^{(1)}\cdot\Delta{S}_n^N \\
        & = W_{n-1}^N\Delta{N}_n+\xi^{(1)}_n\cdot(N_{n-1}\Delta{S}^N_n+\Delta{S}_n^N\Delta{N}_n) .
                \end{split}
                \end{equation}
                As $\bar S_n = S^N_nN_n$, we get 
    \begin{equation}\nn
    \bs
       \d \bar S_n& = \d (S^N_n N_n)   =   S^N_n N_n - S^N_{n-1} N_{n-1}  =   S^N_{n-1} \d N_n + N_{n-1} \d S^N_n + \d S^N_n \d N_n,
                \end{split}
                \end{equation}                
                and thus
                $$ N_{n-1} \d S^N_n + \d S^N_n \d N_n = \d \bar S_n - S^N_{n-1} \d N_n .$$
                This allows to rewrite (particularly, the last term in)  \eqref{875} as 
    \begin{equation}\label{877}
    \bs 
         \xi_n^{(2)}\cdot\Delta{\bar S}_n & = W_{n-1}^N\Delta{N}_n+\xi^{(1)}_n\cdot(\Delta{\bar S}_{n-1}-S^N_{n-1}\Delta{N}_n)\\
        & = \xi_n^{(1)}\cdot\Delta{\bar S}_n+(W_n^N-\xi_n^{(1)}\cdot{S_n^N})\Delta{N}_n \\
        & =\xi_n^{(1)}\cdot\Delta{\bar S}_n+\(V_0+\sum^n_{k=1}\xi^{(1)}_k\cdot\Delta{S}^N_k-\xi^{(1)}_n\cdot{S}^N_n\)\Delta{N}_n.
            \end{split}
                \end{equation} 
    Notice that the self-financing condition for $\xi^{(1)}$, particularly in the form  \eqref{873}, implies that 
    \begin{equation*}
        V_0+\sum^n_{k=1}\xi^{(1)}_k\cdot\Delta{S}^N_k-\xi^{(1)}_n\cdot{S}^N_n=0,
    \end{equation*}
    which is precisely the term in the parentheses in the last line of \eqref{877}. This allows to rewrite \eqref{877}  as \eqref{878}. We obtain \eqref{879} similarly. 
 
\end{proof}

Next, we apply Lemmas \ref{lem871} and \ref{lem872} to characterize 
F{\"o}llmer-Schweizer decomposition under a change of  num\'eraire. 
For simplicity of notations, for random variables $X$ and $Y$, let us introduce 
\begin{equation}\label{3.6}
    \mathcal{C}^N_{\mathcal{F}_n}(X,Y) := \mathbb{E}_{\mathcal{F}_n} \left[ \left( \frac{X}{N_{T}} - \frac{\mathbb{E}_{\mathcal{F}_n} [XN_{T}^{-2}]}{N_{T} \mathbb{E}_{\mathcal{F}_n} [N_{T}^{-2}]} \right) \left( \frac{Y}{N_{T}} - \frac{\mathbb{E}_{\mathcal{F}_n} [YN_{T}^{-2}]}{N_{T} \mathbb{E}_{\mathcal{F}_n} [N_{T}^{-2}]} \right)\right],\quad n\in\{1, \ldots, T\}, 
\end{equation}
and consider the following matrix-valued process 
\be\label{defC}
\bs
    &\mathbfcal{C}_n :=(\mathcal{C}^{N}_{\mathcal{F}_{n-1}}(\Delta S_n^i, \Delta S_n^j))_{i=1,\ldots, d,j=1,\ldots, d},\quad  n \in \{1,...,T\}.
\end{split}
\ee
If $\mathbfcal{C}_n$ is invertible with probability $1$ for all $n$, we can define recursively, backward  in time, vector-valued processes $\xi$ and $\mathbf{c}$ as follows 
\be\label{3.7}
\bs
\xi_{n} &: = [\mathbfcal{C}_{n}]^{-1}  \mathbf{c}_{n} ,\quad n\in\{T,\ldots, 1\},\quad where \\
\mathbf{c}_n &:= \left\{
\begin{array}{lcl}
        (\mathcal{C}_{\mathcal{F}_{T-1}} ^{N}(H, \Delta S_T^i))_{i=1,\ldots, d},  &if& n =T,\\
        (\mathcal{C}_{\mathcal{F}_{n-1}} ^{N}(H - \sum_{k=n+1}^T \xi_k \cdot \Delta S_k, \Delta S_n^i))_{i=1,\ldots, d}, &if& n \in \{T-1,\ldots,1\}.
        \end{array}
        \right.
\end{split}
\ee
For $\xi$ given by \eqref{3.7}, let us also set 
\be\label{eqVn}\bs
    V_n & = \mathbb{E}_{\mathcal{F}_{n}}\[\(H-\sum_{k=n+1}^T \xi_k \cdot \Delta S_k\)  
   \frac{ N_{T}^{-2}}{\mathbb{E}_{\mathcal{F}_{n}}[N_{T}^{-2}]}\],\quad n \in \{0,\ldots,T-1\}.
    \end{split}
\ee
\begin{Theorem}\label{thm3.3}
Consider a model with $T$ periods and $d$ stocks. Let us consider $\mathbfcal{C}$, defined in \eqref{defC}, and assume that   $[\mathbfcal{C}_{n}]^{-1}$ exists for every $n  \in \{1,\ldots,T\}$. 
Then $V_0$, defined through \eqref{eqVn}, is the fair price, and $\xi$, defined in \eqref{3.7}, is the (fair price-based) hedging strategy under the num\'eraire $N$, that is the optimizers to \eqref{3.5}. We also have
\be\label{881}
    \frac{H}{N_T} =  V_0 + \sum_{k=1}^T \xi_k \cdot \Delta S^{N}_k +  L_T,
\ee
where $L$ is the unhedgeable part that satisfies properties $(1)$ and $(2)$ in the statement of Theorem \ref{thm2.2}.
\end{Theorem}
\begin{Remark}
The process $V$ defined in \eqref{eqVn} can be thought as the {\it conditional fair price process under the num\'eraire $N$}, where $V_n$, $n\in\{0,\ldots,T-1\}$,  is the fair price under the information available up to time $n$. We can further extend $V$ to $T$, by setting $V_T = \frac{H}{N_T}$.  In particular, at $n = 0$, $V_0$ is the conditional fair price under trivial information represented by $\cF_0$. This is consistent with the (usual, or rather unconditional) fair price.  
\end{Remark}
\begin{Remark}
Invertibility of  $\mathbfcal{C}_n$'s for every $n  \in \{1,\ldots,T\}$ is closely connected to the non-redundancy of  $d$ stocks. 
\end{Remark}
\begin{proof}[Proof of Theorem \ref{thm3.3}]
    
The proof proceeds recursively, backward  in time. For brevity of the exposition, we focus on the main step and consider the minimization problem at time $n\in\{0,\ldots,T-1\}$, and if $n + 1<T$, assume we have already found $\xi_{n+2},\ldots,\xi_T$.  
Let us denote  
\be\label{8715}
\tilde V_{n+1} := H - \sum_{k=n+2}^T \xi_k \cdot \Delta S_k.
\ee
Now, 
at time $n$, we want to minimize 
\begin{equation}\label{3.8}
    \mathbb{E}_{\mathcal{F}_n} \left[ \bigg( \frac{\tilde V_{n+1} - V_n - \xi_{n+1} \cdot \Delta S_{n+1}}{N_{T}} \bigg)^2 \right],  
\end{equation}
where the minimization is taken over all random variables $V_n$ and $\xi_{n+1}$ measurable with respect to $\cF_{n}$.

Using the first-order conditions, we take the partial derivative of the objective function in \eqref{3.8} with respect to $V_n$, to obtain 
\begin{equation}\nn
\bs
    \frac{\partial}{\partial{V_n}}  \mathbb{E}_{\mathcal{F}_{n}} \left[ \bigg( \frac{\tilde V_{n+1} - V_n - \xi_{n+1} \cdot \Delta S_{n+1}}{N_{T}} \bigg)^2 \right] =-2\mathbb{E}_{\mathcal{F}_{n}} \left[ \frac{\tilde V_{n+1}-V_n-\xi_{n+1} \cdot \Delta{S}_{n+1}}{N_{T}^2} \right]=0,& 
     \end{split}
\end{equation}
and therefore, we get 
$$ V_n = \frac{\mathbb{E}_{\mathcal{F}_{n}}[\tilde V_{n+1} N_T^{-2}] - \xi_{n+1} \cdot \mathbb{E}_{\mathcal{F}_{n}}[\Delta S_{n+1} N_{T}^{-2}]}{\mathbb{E}_{\mathcal{F}_{n}}[N_{T}^{-2}]},$$
which is, in view of \eqref{8715}, is exactly \eqref{eqVn}. 
Next we substitute this $V_n$ back into our objective function \eqref{3.8} and take partial derivatives with respect to each component of $\xi$, to obtain
\begin{equation*}\bs
\frac{\partial}{\partial{\xi^j}} \mathbb{E}_{\mathcal{F}_n} \left[ \bigg( \frac{\tilde V_{n+1}}{N_{T}} - \frac{\mathbb{E}_{\mathcal{F}_n}[\tilde V_{n+1}N_{T}^{-2}]}{N_{T} \mathbb{E}_{\mathcal{F}_n}[N_{T}^{-2}]} - \xi_{n+1} \cdot \bigg( \frac{\Delta S_{n+1}}{N_{T}} - \frac{\mathbb{E}_{\mathcal{F}_n}[\Delta S_{n+1} N_{T}^{-2}]}{N_{T} \mathbb{E}_{\mathcal{F}_n}[N_{T}^{-2}]} \bigg) \bigg)^2 \right] = 0,\\
j \in \{1,\ldots,d\}.
\end{split}
\end{equation*}
Computing these derivatives and using the notation specified in \eqref{3.6}, we find that
\begin{equation*}
    \sum_{i=1}^d \xi^i_{n+1} \mathcal{C}^N_{\mathcal{F}_n}(\Delta S^i, \Delta S^j) = \mathcal{C}^N_{\mathcal{F}_n} (\tilde V_{n+1}, \Delta S^j),\quad  j \in \{1,...,d\}.
\end{equation*}
Recalling the notations for $\mathbfcal{C}$ in \eqref{defC} and for $ \mathbf{c}$ in \eqref{3.7}, we can rewrite the latter equation as 
\begin{equation*}
    \mathbfcal{C}_{n+1}   \xi_{n+1} = \mathbf{c}_{n+1}.
\end{equation*}
Now, using the assumed invertibility of $\mathbfcal{C}_{n+1}$, we have  
\begin{equation*}
    \xi_{n+1} = [\mathbfcal{C}_{n+1}]^{-1}   \mathbf{c}_{n+1},
\end{equation*}
which gives $\xi$ in \eqref{3.7}. Now, from Lemmas \ref{lem871} and \ref{lem872}, one can now show \eqref{881}. 
\end{proof}

\section{Examples}\label{secExam}
To illustrate the results  and to highlight some special features related to fair pricing under the change of num\'eraire,  we consider the following examples, where we find the optimal trading strategy and the corresponding fair price 
using Theorem \ref{thm3.3}. Key features of the results can be illustrated in a one-period trinomial model with one risky asset. Let the initial stock price $S_0 = 2$. Also let an increase in stock price happen by the factor of $u=2$, a lack of movement be $c=1$, and the down movement occur by the factor $d=1/2$ so that $uS_0=4$, $cS_0=2$, and $dS_0=1$ and so on. 
Let the probability of an up move occurring be $\frac{1}{6}$, the probability of a down move be $\frac{1}{3}$, and the probability of the stock staying steady be $\frac{1}{2}$. 
\begin{figure}[h]\label{fig1}
 \centering
    \begin{tikzpicture}[line cap=round,line join=round,>=triangle 45,x=1cm,y=1cm];
    \draw[color=black] (0, 0) [xshift=-17pt] node {$S_0 = 2$};
    \draw [-stealth](0,0) -- (2,1);
    \draw[color=black] (2, 1) [xshift= 20pt] node {$uS_0 = 4$};
    \draw [-stealth](0,0) -- (2,0);
    \draw[color=black] (2, 0) [xshift= 20pt] node {$cS_0 = 2$};
    \draw [-stealth](0,0) -- (2,-1);
    \draw[color=black] (2, -1) [xshift= 20pt] node {$dS_0 = 1$};
    
    \end{tikzpicture}
    \caption{Figure 1. One-period Trinomial Model.}
    \label{trinomial}
\end{figure}
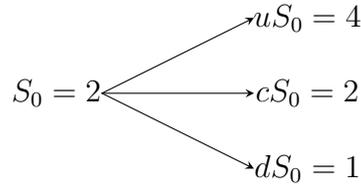    


The following example shows that for non-replicable contingent claims, the fair price and the optimal strategy in the sense of   optimization problem \eqref{3.5} change with $N$.  

\begin{Example}\label{ex1}
We demonstrate the results of section 3 as applied to the trinomial model, obtaining the optimal hedging strategy and initial capital once a change of num{\'e}raire has been enacted. 

(a) Consider a tradeable num{\'e}raire given by $1+\frac{1}{2}\Delta{S}$. Let $H$ be a European call option which yields $\max\{0,K-S_{1}\}$ where $S_{1}$ is the value of the stock at time one, $K$ is the strike price, and $K=3$. Then, using Theorem \ref{thm3.3} for the one-period  one-stock case, we deduce 
\begin{equation*}
\xi_1
=\frac{\mathcal{C}^N_{\cF_0}(H,\Delta{S}_1)}{\mathcal{C}^N_{\cF_0}(\Delta{S}_1, \d S_1)}
=\frac{\mathbb{E}\[\frac{\(H\mathbb{E}\[N^{-2}_1\]-\mathbb{E}\[HN^{-2}_1\]\)\(\Delta{S}_1\mathbb{E}\[N^{-2}_1\]-\mathbb{E}[\Delta{S}_1N^{-2}_1]\)}{\(N_1\mathbb{E}\[N^{-2}_1\]\)^2}\]}
{\mathbb{E}\[\(\frac{\Delta{S_1}\mathbb{E}\[N^{-2}_1\]-\mathbb{E}\[\Delta{S}_1N^{-2}_1\]}{N_1\mathbb{E}\[N^{-2}_1\]}\)^2\]}.
\end{equation*}
Computing the expectations using the values above  gives
\begin{equation*}
\E[N^{-2}_1]
=\left(\frac{1}{1+\frac{1}{2}(2)}\right)^{2}\frac{1}{6}
+\left(\frac{1}{1+\frac{1}{2}(0)}\right)^2\frac{1}{2}
+\left(\frac{1}{\frac{1}{2}}\right)^2\frac{1}{3}
=\frac{45}{24}.
\end{equation*}
Substitution yields 
\begin{equation*}
    \xi=\frac{75}{576}\approx0.13021.
\end{equation*}
This would indicate that the optimal amount of stock to buy at time zero is approximately 0.13021 of a share. The same process can be used to find $V_0$. Calculating the expectations gives 
\begin{equation*}
    V_0 = \frac{471}{4320} \approx 0.10903 .
\end{equation*}

(b) {\it Comparison with Optimal Strategy when $N\equiv 1$.}
For comparison, we now assume the num{\'e}raire $N\equiv 1$ and use the original one-period, one-stock formulas for the optimal trading strategy and initial value. The choice of  $N\equiv 1$ simplifies the formulas for the optimal trading strategy $\xi$ and the fair price $V_0$ to the following equations for the one-period, one-stock case:
\begin{equation}
\label{879}
\xi_1
=\frac{\text{Cov}(H,\Delta{S}_1)}{\text{Var}(\Delta{S}_1)}
=\frac{\mathbb{E}\[(H-\mathbb{E}[H])(\Delta{S}_1-\mathbb{E}[\Delta{S}_1])\]}
{\mathbb{E}[(\Delta{S}_1-\mathbb{E}[\Delta{S}_1])^2]}\quad and \quad 
    V_0=\mathbb{E}[H]-\xi_1\mathbb{E}[\Delta{S}_1].
\end{equation}
Using the same method and information described above, the optimal trading strategy under the num{\'e}raire $N\equiv 1$ is $\frac{1}{3}$ and the fair price is $\frac{1}{6}$. Notice that this differs from the optimal strategy and fair price under a different num{\'e}raire. 
Thus, the choice of the num\'eraire affect the optimal the optimal strategy and the fair price. 
\end{Example}

The following example shows that changes of num\'eraire are different from the changes of interest rates, in general. 

\begin{Example}

There is the relatively common supposition that a change in the num{\'e}raire relates to a change in the interest rate, or at the very least, that perturbations of both of them act similarly. In this example, we include calculations on the interest rate in order to lay the groundwork for the understanding that perturbations of the interest rate and perturbations of the num{\'e}raire are not related, in general. In one-period settings, let us consider the model of the stock price as in Example \ref{ex1}, and let us suppose that the interest rate is a constant $r>0$ (instead of 0, as in  part (b) of Example \ref{ex1}). Formulating the optimization problem similar to \eqref{3.5}, where the role of $N$ is played by the bank account, leads to the computations performed  in  \cite[Section 2]{FS89}), which assert that the optimal $\xi$ does not change compared to \eqref{879}, whereas the fair price $V_0$ is given by  
\begin{equation*}
    V_0=\mathbb{E}\[\frac{H}{1+r}\]-\xi_1 \mathbb{E}\[\frac{S_1}{1+r}-S_0\],
\end{equation*}
which is different from $V_0$ specified through Theorem \ref{thm3.3}, in general, even notationally, as $r$ does not enter \eqref{3.5}. This simple example shows that the perturbations of the interest rate are different from the perturbations of the num\'eraire. 
\end{Example}
\section{Stability under Perturbations of the Num{\'e}raire}\label{secStab}

We now address the stability of the F{\"o}llmer-Schweizer decomposition under perturbations of the num{\'e}raire. Stability has already been shown for perturbations of $V_T$ in \cite{monatStricker} and for perturbations of $S$ in \cite{MostovyiREU2019}. 
We   consider a family of $\mathbb{F}$-adapted strictly positive num{\'e}raire processes parameterized by $\epsilon$, writing $(N^\epsilon)_{\epsilon \in (-\epsilon_0, \epsilon_0)}$ for some $\epsilon_0 > 0$. A specific example of such a family is given in the following section, but for now we only suppose that
\begin{equation}\label{5.1}
    \lim_{\epsilon \to 0} N^\epsilon_n(\omega) = N^0_n(\omega) = 1 \hspace{3mm} \text{for every} \hspace{2mm} n \in \{0,...,T\} \hspace{2mm} \text{and} \hspace{2mm} \omega \in \Omega.  
\end{equation}
\begin{Remark}
We stress that, 
in \eqref{5.1} and below, the limits should be understood in the following sense. First, we fix $n$ and $\omega$, then we take a limit as $\e \to 0$. Thus, the limit in \eqref{5.1} and other limits below, hold for every $\omega\in\Omega$. In particular, the set of $\omega$'s, for which \eqref{5.1} and other limits below exist, has probability $1$.  
\end{Remark}
\begin{Theorem}\label{thm5.1}
For some $\epsilon_0>0$, let us consider a family of num{\'e}raire processes\footnote{Strict positivity for each $\epsilon \in (-\epsilon_0, \epsilon_0)$ is embedded in the definition of the num\'eraires.} \newline
$((N^\epsilon_n)_{n \in \{0,...,T\}})_{\epsilon \in (-\epsilon_0, \epsilon_0)}$ satisfying \eqref{5.1}.  Let us suppose that,  for every $n \in \{1,...,T\}$,  $\({Cov}_{\mathcal{F}_{n-1}}(\Delta S_n^i, \Delta S_n^j)\)_{i=1,\ldots, d,j=1,\ldots,d}$, is invertible  with probability $1$\footnote{This conditional covariance matrix-valued process is  $\mathbfcal{C}$ defined in \eqref{defC} for $N^0 \equiv 1$.}. Then there exists $\bar\e_0>0\in(0, \e_0]$, such that  for every random variable $H$ and $\e\in(-\bar\e_0, \bar\e_0)$, the assumptions of Theorem \ref{thm3.3} are satisfied. Further, the corresponding family of num{\'e}raire adjusted F{\"o}llmer-Schweizer  decompositions\footnote{These decompositions are  given via \eqref{881} in Theorem \ref{thm3.3}.}
\begin{equation*}
    \frac{H}{N^\epsilon_T} = V_0^\epsilon + \sum_{k=1}^T \xi_k^\epsilon \cdot \Delta S_k^{N^\epsilon} + L^\epsilon_T, \hspace{3mm} \epsilon \in (-\bar\epsilon_0, \bar\epsilon_0),
\end{equation*}
 and where $\xi^\e$'s are given via  \eqref{3.7},
satisfies
\begin{align*}
    &\lim_{\epsilon \to 0} \xi_n^\epsilon  = \xi_n^0, \hspace{3mm} n \in \{1,...,T\}, \\
    &\lim_{\epsilon \to 0} \Delta S_n^{N^\epsilon}  = \Delta \bar S_n, \hspace{3mm} n \in \{1,...,T\} ,\\
    &\lim_{\epsilon \to 0} V_0^\epsilon  = V_0^0, \\
    &\lim_{\epsilon \to 0} L_n^\epsilon  = L_n^0, \hspace{3mm} n \in \{0,...,T\}.
\end{align*}
\end{Theorem}
\begin{proof}
The proof goes recursively, backward  in $n$. First, consider $n=T$. Notice that since we are working on a finite probability space, via the definition of conditional expectation, from \eqref{5.1}, we get
\begin{equation*}
    \lim_{\epsilon \to 0} \mathbb{E}_{\mathcal{F}_{T-1}}[N^\epsilon_T] = \mathbb{E}_{\mathcal{F}_{T-1}}[N^0_T] = 1,
\end{equation*}
and by continuity of $f(x)=x^{-2}$ on $(0, \infty)$, we obtain  
\begin{equation*}
    \lim_{\epsilon \to 0} \mathbb{E}_{\mathcal{F}_{T-1}}[(N^\epsilon_T)^{-2}] = \mathbb{E}_{\mathcal{F}_{T-1}}[(N^0_T)^{-2}] = 1.
\end{equation*}
Moreover, this implies for random variables $X$ and $Y$, both not depending on $\epsilon$, we have 
\be\label{5.2}
\bs
    \lim_{\epsilon \to 0} \mathcal{C}^{N^\epsilon_T}_{\mathcal{F}_{T-1}}[X,Y] & = \lim_{\epsilon \to 0} \mathbb{E}_{\mathcal{F}_{T-1}}\left[\left(\frac{X}{N^\epsilon_T} - \frac{\mathbb{E}_{\mathcal{F}_{T-1}}[X(N^\epsilon_T)^{-2}]}{N^\epsilon_T \mathbb{E}_{\mathcal{F}_{T-1}}[(N^\epsilon_T)^{-2}]} \right) \left(\frac{Y}{N^\epsilon_T} - \frac{\mathbb{E}_{\mathcal{F}_{T-1}}[Y(N^\epsilon_T)^{-2}]}{N^\epsilon_T \mathbb{E}_{\mathcal{F}_{T-1}}[(N^\epsilon_T)^{-2}]} \right) \right] \\
    & = \mathbb{E}_{\mathcal{F}_{T-1}}\left[\lim_{\epsilon \to 0}\left(\frac{X}{N^\epsilon_T} - \frac{\mathbb{E}_{\mathcal{F}_{T-1}}[X(N^\epsilon_T)^{-2}]}{N^\epsilon_T \mathbb{E}_{\mathcal{F}_{T-1}}[(N^\epsilon_T)^{-2}]} \right) \left(\frac{Y}{N^\epsilon_T} - \frac{\mathbb{E}_{\mathcal{F}_{T-1}}[Y(N^\epsilon_T)^{-2}]}{N^\epsilon_T \mathbb{E}_{\mathcal{F}_{T-1}}[(N^\epsilon_T)^{-2}]} \right) \right] \\
    & = \mathbb{E}_{\mathcal{F}_{T-1}}\bigg[\left(\frac{X}{\lim_{\epsilon \to 0}N^\epsilon_T} - \frac{\lim_{\epsilon \to 0}\mathbb{E}_{\mathcal{F}_{T-1}}[X(N^\epsilon_T)^{-2}]}{\lim_{\epsilon \to 0}N^\epsilon_T \lim_{\epsilon \to 0}\mathbb{E}_{\mathcal{F}_{T-1}}[(N^\epsilon_T)^{-2}]} \right) \\ & \hspace{1.68cm} \left(\frac{Y}{\lim_{\epsilon \to 0}N^\epsilon_T} - \frac{\lim_{\epsilon \to 0}\mathbb{E}_{\mathcal{F}_{T-1}}[Y(N^\epsilon_T)^{-2}]}{\lim_{\epsilon \to 0}N^\epsilon_T \lim_{\epsilon \to 0}\mathbb{E}_{\mathcal{F}_{T-1}}[(N^\epsilon_T)^{-2}]} \right) \bigg]\\
    & = \mathbb{E}_{\mathcal{F}_{T-1}}\left[\left(\frac{X}{N^0_T} - \frac{\mathbb{E}_{\mathcal{F}_{T-1}}[X(N^0_T)^{-2}]}{N^0_T \mathbb{E}_{\mathcal{F}_{T-1}}[(N^0_T)^{-2}]} \right) \left(\frac{Y}{N^0_T} - \frac{\mathbb{E}_{\mathcal{F}_{T-1}}[Y(N^0_T)^{-2}]}{N^0_T \mathbb{E}_{\mathcal{F}_{T-1}}[(N^0_T)^{-2}]} \right) \right]\\
    & = \mathbb{E}_{\mathcal{F}_{T-1}}[(X - \mathbb{E}_{\mathcal{F}_{T-1}}[X])(Y - \mathbb{E}_{\mathcal{F}_{T-1}}[Y])]\\
    & = \text{Cov}_{\mathcal{F}_{T-1}}[X,Y],  
    \end{split}
\ee
which holds for every $\omega \in \Omega$. 
Thus, the continuity of $ \mathcal{C}^{N^\epsilon_T}_{\mathcal{F}_{T-1}}[X,Y] $ in $\e$, and invertibility of $\mathcal{C}^{N^0_n}$, for every $n\in\{1,\ldots, T\}$, imply that there exists $\bar \e_0\in(0, \e_0)$, such that $\mathcal{C}^{N^\e_n}$ is invertible for every $\e\in(-\bar \e_0, \bar \e_0)$ and every $n\in\{1,\ldots, T\}$. Such invertibility also implies continuity of $ [(\mathcal{C}^{N^\epsilon_T}_{\mathcal{F}_{T-1}}[\Delta S^i_T,\Delta S^j_T])_{i=1,j=1}^d]^{-1}$ at $\e = 0$, that is $$\lim_{\epsilon \to 0}[(\mathcal{C}^{N^\epsilon_T}_{\mathcal{F}_{T-1}}[\Delta S^i_T,\Delta S^j_T])_{i=1,j=1}^d]^{-1} = 
[(\mathcal{C}^{N^0_T}_{\mathcal{F}_{T-1}}[\Delta S^i_T,\Delta S^j_T])_{i=1,j=1}^d]^{-1}.$$
Consequently, from \eqref{5.2}, we deduce that 
 \be\label{5.3}
\bs
    \lim_{\epsilon \to 0} \xi^\epsilon_T & = \lim_{\epsilon \to 0}[(\mathcal{C}^{N^\epsilon_T}_{\mathcal{F}_{T-1}}[\Delta S^i_T,\Delta S^j_T])_{i=1,j=1}^d]^{-1}(\mathcal{C}^{N^\epsilon_T}_{\mathcal{F}_{T-1}}[H, \Delta S_T^i])_{i=1}^d \\
    & = [(\lim_{\epsilon \to 0} \mathcal{C}^{N^\epsilon_T}_{\mathcal{F}_{T-1}}[\Delta S^i_T,\Delta S^j_T])_{i=1,j=1}^d]^{-1}(\lim_{\epsilon \to 0} \mathcal{C}^{N^\epsilon_T}_{\mathcal{F}_{T-1}}[H, \Delta S_T^i])_{i=1}^d \\
    & = [(\text{Cov}_{\mathcal{F}_{T-1}}[\Delta S^i_T,\Delta S^j_T])_{i=1,j=1}^d]^{-1}(\text{Cov}_{\mathcal{F}_{T-1}}[H, \Delta S_T^i])_{i=1}^d \\
    & = \xi^0_T  ,
    \end{split}
\ee
for every $\omega \in \Omega$. If $T=1$, this completes the proof for stability of $\xi$. If $T>1$, defining
\begin{equation*}
    A^\epsilon_n := H - \sum_{k=n+1}^T \xi^\epsilon_k \cdot \Delta   S^{N^\e}_k, \hspace{3mm} n \in \{0,...,T-1\},\quad \epsilon \in (-\epsilon_0,\epsilon_0),
\end{equation*}
from \eqref{5.3}, we get
\begin{equation*}
    \lim_{\epsilon \to 0} A^\epsilon_{T-1} = A^0_{T-1}, \hspace{4mm} \omega \in \Omega.
\end{equation*}
Consequently, as with \eqref{5.2}, we obtain 
\begin{align*}
     \lim_{\epsilon \to 0}  \mathcal{C}^{N^\epsilon_{T}}_{\mathcal{F}_{T-2}}[A^\epsilon_{T-1},\Delta S_{T-1}] & =\lim_{\epsilon \to 0}  \mathbb{E}_{\mathcal{F}_{T-2}}\bigg[\left(\frac{A^\epsilon_{T-1}}{N^\epsilon_{T}} - \frac{\mathbb{E}_{\mathcal{F}_{T-2}}[A^\epsilon_{T-1}(N^\epsilon_{T})^{-2}]}{N^\epsilon_{T} \mathbb{E}_{\mathcal{F}_{T-2}}[(N^\epsilon_{T})^{-2}]} \right) \\
     & \hspace{2.26cm}\left(\frac{\Delta S_{T-1}}{N^\epsilon_{T}} - \frac{\mathbb{E}_{\mathcal{F}_{T-2}}[\Delta S_{T-1}(N^\epsilon_{T})^{-2}]}{N^\epsilon_{T} \mathbb{E}_{\mathcal{F}_{T-2}}[(N^\epsilon_{T})^{-2}]} \right) \bigg] \\
     & =  \mathbb{E}_{\mathcal{F}_{T-2}}\bigg[\left(\frac{\lim_{\epsilon \to 0}A^\epsilon_{T-1}}{\lim_{\epsilon \to 0}N^\epsilon_{T}} - \frac{\lim_{\epsilon \to 0}\mathbb{E}_{\mathcal{F}_{T-2}}[A^\epsilon_{T-1}(N^\epsilon_{T})^{-2}]}{\lim_{\epsilon \to 0}N^\epsilon_{T} \mathbb{E}_{\mathcal{F}_{T-2}}[(N^\epsilon_{T})^{-2}]} \right) \\ &  \hspace{1.68cm} \left(\frac{\Delta S_{T-1}}{\lim_{\epsilon \to 0}N^\epsilon_{T}} - \frac{\lim_{\epsilon \to 0}\mathbb{E}_{\mathcal{F}_{T-2}}[\Delta S_{T-1}(N^\epsilon_{T})^{-2}]}{\lim_{\epsilon \to 0}N^\epsilon_{T} \mathbb{E}_{\mathcal{F}_{T-2}}[(N^\epsilon_{T})^{-2}]} \right) \bigg] \\
     & = \mathbb{E}_{\mathcal{F}_{T-2}}[(A^0_{T-1} - \mathbb{E}_{\mathcal{F}_{T-2}}[A^0_{T-1}])(\Delta S_{T-1} - \mathbb{E}_{\mathcal{F}_{T-2}}[\Delta S_{T-1}])] \\
     & = \text{Cov}_{\mathcal{F}_{T-2}}[A^0_{T-1}, \Delta S_{T-1}],\quad     \omega \in \Omega,
\end{align*}
and as with \eqref{5.3}, we obtain the vector equation
\begin{align*}
    \lim_{\epsilon \to 0} \xi^\epsilon_{T-1} & = \lim_{\epsilon \to 0}[(\mathcal{C}^{N^\epsilon_{T}}_{\mathcal{F}_{T-2}}[\Delta S^i_{T-1},\Delta S^j_{T-1}])_{i=1,j=1}^d]^{-1}(\mathcal{C}^{N^\epsilon_{T}}_{\mathcal{F}_{T-2}}[A^\epsilon_{T-1}, \Delta S_{T-1}^i])_{i=1}^d \\
    & = [(\lim_{\epsilon \to 0} \mathcal{C}^{N^\epsilon_{T}}_{\mathcal{F}_{T-2}}[\Delta S^i_{T-1},\Delta S^j_{T-1}])_{i=1,j=1}^d]^{-1}(\lim_{\epsilon \to 0} \mathcal{C}^{N^\epsilon_{T}}_{\mathcal{F}_{T-2}}[A^\epsilon_{T-1}, \Delta S_{T-1}^i])_{i=1}^d \\
    & = [(\text{Cov}_{\mathcal{F}_{T-2}}[\Delta S^i_{T-1},\Delta S^j_{T-1}])_{i=1,j=1}^d]^{-1}(\text{Cov}_{\mathcal{F}_{T-2}}[A^0_{T-1}, \Delta S_{T-1}^i])_{i=1}^d \\
    & = \xi^0_{T-1},\quad  \omega \in \Omega. 
\end{align*}
Proceeding in this manner, one can show
\begin{equation}\label{8711}
    \lim_{\epsilon \to 0} \xi^\epsilon_n = \xi^0_n, \quad  n \in \{1,...,N\}, \quad  \omega \in \Omega.  
\end{equation}
We also obtain 
\begin{equation*}
    \lim_{\epsilon \to 0} \Delta S^{N^\epsilon}_n = \lim_{\epsilon \to 0} \frac{\bar S_{n}}{N^\epsilon_{n}} - \frac{\bar S_{n-1}}{N^\epsilon_{n-1}} = \frac{\bar S_{n}}{\lim_{\epsilon \to 0}N^\epsilon_{n}} - \frac{\bar S_{n-1}}{\lim_{\epsilon \to 0}N^\epsilon_{n-1}} = \bar S_n -\bar  S_{n-1} = \Delta \bar S_n,
\end{equation*}
giving us via \eqref{8711} the following equality 
\begin{equation}\label{8712}
    \lim_{\epsilon \to 0} \sum_{k=1}^n \xi_k^\epsilon \cdot \Delta S_k^{N^\epsilon} = \sum_{k=1}^n \xi_k^0 \cdot \Delta \bar S_k , \quad  n \in \{1,\ldots, T\}, \hspace{3mm} \omega \in \Omega. 
\end{equation}
Therefore, by taking expectation in
\begin{equation}\label{8713}
    \frac{H}{N^\epsilon_T} = V_0^\epsilon + \sum_{k=1}^T \xi_k^\epsilon \cdot \Delta S_k^{N^\epsilon} + L^\epsilon_T,
\end{equation}
and using $\lim_{\epsilon \to 0} \frac{H}{N^\epsilon} = H$, $\mathbb{E}[L^\epsilon_T]=0$, from \eqref{8712}, we get $\lim_{\epsilon \to 0} V_0^\epsilon = V^0_0$. Consequently,  from convergence of the left hand side in \eqref{8713} to $H$, convergence of $V^\epsilon_0$ to $V_0^0$, and \eqref{8712}, we obtain $\lim_{\epsilon \to 0} L^\epsilon_T = L^0_T$. Finally, using the martingale condition on $L^\epsilon$, $\mathbb{E}_{\mathcal{F}_n}[L^\epsilon_T] = L_n^\e$, we conclude the proof with
\begin{equation*}
    \lim_{\epsilon \to 0} L^\epsilon_n = L^0_n,\quad n\in\{0,\ldots, T\}.
\end{equation*}
\end{proof}
\section{Asymptotic Analysis}\label{secAsym}
In order to quantify how the fair price and trading strategy respond to num\'eraire perturbations, we introduce a (linear) parameterization of a tradable numeraire given by
\begin{equation}\label{6.1}
    N^\epsilon_n = 1 + \epsilon \sum_{k=1}^n \eta_k \cdot \Delta S_k, \quad  \epsilon \in (-\epsilon_0, \epsilon_0),
\end{equation}
where $\eta\in\Theta$, and $\epsilon_0$ is chosen so that $N^\epsilon > 0$ for every $\epsilon \in(-\e_0, \e_0)$. Note that this satisfies \eqref{5.1}, and so Theorem \ref{thm5.1} is used throughout this section. 
Now we define the following processes  
\begin{equation*}
    N'_n := \lim_{\epsilon \to 0} \frac{N^\epsilon_n-N^0_n}{\epsilon} = \lim_{\epsilon \to 0} \frac{ 1 + \epsilon \sum_{k=1}^n \eta_k \cdot \Delta S_k - 1}{ \epsilon} = \sum_{k=1}^n \eta_k \cdot \Delta S_k,\quad n\in\{0,\dots, T\}.
\end{equation*}
For every  $\e \in (-\e_0, \e_0)$, we denote by $\mathbfcal{C}^\e$, $\mathbf{c}^\e$, and $\xi^\e$ the processes defined in \eqref{defC} and \eqref{3.7}, respectively, corresponding to the num\'eraire $N^\e$. We also set   
$$J'_n := -2N'_T + 2\E_{\cF_n}\[N'_T\],$$
\be\nn\bs
\cC'_{n}(X,Y): = 
&-\E_{\cF_n}\[XJ'_n (Y -\E_{\cF_n}[Y])\] - \E_{\cF_n}\[YJ'_n (X -\E_{\cF_n}[X])\] \\
&- 2\E_{\cF_n}[(X -\E_{\cF_n}[X])N'_T(Y -\E_{\cF_n}[Y])],
\end{split}
\ee
\be\nn
\mathbfcal{C}'_n :=\left( \cC'_{n-1}(\d S^i_n,\d S^j_n) \right)_{i \in\{ 1,\ldots, d\}, j \in\{ 1,\ldots, d\}},\quad n\in\{1,\ldots,T\}.
\ee
For $n = T$, we introduce 
\be
\mathbf{c}'_T := 
  \(\cC'_{T-1}(H,\d S^i_T)\)_{i \in\{ 1,\ldots, d\}}, 
\ee
\be
\label{eqxi'T}
\xi'_T : = [\mathbfcal{C}^0_T]^{-1}\( \mathbf c'_T - \mathbfcal{C}'_T\xi^0_T\).
\ee
Continuing  recursively, backward in time, for every $n\in\{T-1,\ldots, 1\}$, we define  
\be\nn
A_{n+1}' = -\sum\limits_{k = n+1 }^T \xi'_k \cdot \d \bar S_k + \sum\limits_{k = n  +1}^T\xi^0_k \cdot \d(\bar SN')_k,
\ee
where $ \bar SN' $ is a vector-valued stochastic process  $  (N', S^1 N',\ldots, S^dN')$,
\be\nn\bs
\tilde \cC'_{n}(X,Y): = \cC'_{n}(X,Y) + \E_{\cF_n}\[ (A'_{n+1} - \E[A'_{n+1}])(Y -\E_{\cF_n}[Y])\],
\end{split}
\ee
\be
\mathbf{c}'_n :=  \(\tilde \cC'_{n-1}(H-\sum\limits_{k = n+1 }^T \xi^0_k \cdot \d \bar S_k,\d S^i_T)\)_{i \in\{ 1,\ldots, d\}}, 
\ee
and  
\be
\label{eqxi'}
\xi'_n : = [\mathbfcal{C}^0_n]^{-1}\( \mathbf c'_n - \mathbfcal{C}'_n\xi^0_n\),\quad n\in\{T-1,\ldots, 1\}.
\ee
The following theorem gives the first-order corrections to the fair price, the  hedging strategy and the unhedgeable component under small perturbations of the num\'eraire. 
\begin{Theorem}
Suppose that a family of num{\'e}raire processes $((N^\epsilon_n)_{n \in \{0,...,T\}})_{\epsilon \in (-\epsilon_0, \epsilon_0)}$ is given by \eqref{6.1}. Let us suppose that $\({Cov}_{\mathcal{F}_{n-1}}(\Delta S_n^i, \Delta S_n^j)\)_{i=1,\ldots, d,j=1,\ldots,d}$, is invertible for every $n \in \{1,...,T\}$  with probability $1$\footnote{This condition is the same as  in Theorem \ref{thm5.1}. Again, we only impose it for the base model corresponding to $\e = 0$.}.  Then for every $H$, there exists $\bar \e_0\in(0, \e_0]$, such that, for  every $\e\in(-\bar \e_0, \bar \e_0)$,  with~probability~1, we have 
\begin{equation}\label{fse}
    \frac{H}{N^\epsilon_T} =  V_0^\epsilon + \sum_{k=1}^T \xi^\epsilon_k \cdot \Delta S^{N^\epsilon}_k +  L^\epsilon_T .
\end{equation}
where $\xi^\e$'s are given via \eqref{3.7} with $N = N^\e$'s.
The first-order corrections to the optimal trading strategy $\xi_n$, fair price $V_0$, and unhedgeable component $L_n$ 
are given by
\be\label{eqxi'f}
 \lim_{\epsilon \to 0} \frac{\xi_n^\epsilon-\xi^0_n}{\epsilon} = \xi'_n,\quad n\in\{1,\ldots, T\},
\ee
where $\xi'_n$, $n\in\{1, \ldots, T-1\}$ is given by \eqref{eqxi'} and $\xi_T'$ is specified in \eqref{eqxi'T},
\be\label{eqV'}
 \lim_{\epsilon \to 0} \frac{V_0^\epsilon-V^0_0}{\epsilon} =
 \E\[ \sum\limits_{k=1}^T \xi^0_k \cdot\d(\bar S N')_k - H_T N'_T -\sum\limits_{k=1}^T \xi'_k\cdot \d \bar S_k  \],
\ee 
\be\label{eqL'}
\bs
 \lim_{\epsilon \to 0} \frac{L_n^\epsilon-L^0_n}{\epsilon} &=
 \E_{\cF_n}\[ \sum\limits_{k=1}^T \xi^0_k\cdot \d(\bar S N')_k\] -\E\[ \sum\limits_{k=1}^T \xi^0_k\cdot \d(\bar S N')_k\]  \\
 &
 -\(\E_{\cF_n}\[ H_T N'_T\] -\E \[ H_T N'_T\] + \E_{\cF_n}\[ \sum\limits_{k=1}^T \xi'_k\cdot\d \bar S_k  \] -  \E \[ \sum\limits_{k=1}^T \xi'_k\cdot \d \bar S_k  \]\),\\
&\hspace{100mm} n\in\{0,\ldots, T\}.
\end{split}
\ee

%
%

\end{Theorem}
\begin{proof}
 The proof parallels the proof of \cite[Theorem 6.3]{MostovyiREU2019}, so, for brevity of the exposition, we only outline the main steps. We observe that invertibility of  $[\mathbfcal{C}^0_{n}]^{-1}$, $n  \in \{1,\ldots,T\}$, and the argument in the proof of Theorem \ref{thm5.1} imply that there exists   $\bar \e_0\in(0, \e_0]$, such that   $[\mathbfcal{C}^\e_{n}]^{-1}$ are invertible for every $n  \in \{1,\ldots,T\}$ and $\e\in(-\bar\e_0, \bar\e_0)$. This implies that the assertions of Theorem \eqref{thm3.3} apply for every $\e\in(-\bar\e_0, \bar\e_0)$, and therefore \eqref{fse} holds. 
To show   \eqref{eqxi'f}, we proceed recursively, backward in time, where \eqref{eqxi'f} follows from direct computations. 

Now, we show \eqref{eqV'}.
Let us consider \eqref{fse}, for every $\e\in(-\bar\e_0, \bar\e_0)$. As $\E[L^\e_T] = 0$, for every $\e\in(-\bar\e_0, \bar\e_0)$, taking the expectation in \eqref{fse}, we deduce that
\be\label{891} V^\e_0  = \E\[\frac H{N^\e_T} - \sum\limits_{k = 1}^T \xi^\e_k \cdot \d S^{N^\e}_k \] 
,\quad  \e\in(-\bar\e_0, \bar\e_0)
.\ee
One can see that 
\be\label{892} 
 \lim\limits_{\epsilon \to 0}\frac 1{\e}\(\E\[\frac H{N^\e_T}\] - \E\[\frac H{N^0_T}\]\) = \E\[ H{N'_T}\],\ee
and
\be\label{893} 
\lim\limits_{\epsilon \to 0}\frac 1{\e} \E\[\sum\limits_{k = 1}^T \xi^\e_k \cdot \d S^{N^\e}_k -  \sum\limits_{k = 1}^T \xi^0_k \cdot \d S^{N^0}_k\] = 
\E\[\sum\limits_{k = 1}^T \xi'_k \cdot \d\bar S_k
-\sum\limits_{k = 1}^T \xi^0_k \cdot \d(\bar S N')_k\],
\ee
Therefore, using \eqref{892} and \eqref{893}, with \eqref{891}, we deduce  that \eqref{eqV'} holds. 

Finally, we show \eqref{eqL'}. Again, we start from \eqref{fse}, which we can rewrite as 
\be\label{896}
L^\e_T = \frac{H}{N^\e_T} - V_0 - \sum\limits_{k = 1}^T \xi^\e_k\cdot \d S^{N^\e}_k.
\ee
and since $L^\e$ is a $\P$-martingale, for every $\e\in(-\bar\e_0, \bar\e_0)$, from \eqref{896}, we obtain \be\label{897}
L^\e_n = \E_{\cF_n}\[L^\e_T\] = - V_0 +  \E_{\cF_n}\[\frac{H}{N^\e_T} - \sum\limits_{k = 1}^T \xi^\e_k\cdot \d S^{N^\e}_k\].
\ee
For every $\omega\in\Omega$ and $n\in\{0, \ldots, T\}$, one can see that  
\be\label{898}
\lim\limits_{\epsilon \to 0}\frac1{\e}\E_{\cF_n}\[\frac{H}{N^\e_T} -\frac{H}{N^0_T}\] = - \E_{\cF_n}\[HN'_T\],
\ee
\be\label{899}
\lim\limits_{\epsilon \to 0}\frac1{\e}\E_{\cF_n}\[\sum\limits_{k = 1}^T \xi^\e_k\cdot \d S^{N^\e}_k-\sum\limits_{k = 1}^T \xi^0_k\cdot \d S^0_k\] =
\E_{\cF_n}\[\sum\limits_{k = 1}^T \xi'_k\cdot \d \bar S_k-\sum\limits_{k = 1}^T \xi^0_k\cdot \d (\bar SN')_k\].
\ee
Therefore, from \eqref{897}, using \eqref{898} and \eqref{899}, we obtain \eqref{eqL'}. 

\end{proof}
\bibliographystyle{alpha}\bibliography{references2021}
\end{document}